\documentclass[letterpaper, 10 pt, conference]{ieeeconf}  

\IEEEoverridecommandlockouts                              

\usepackage[utf8]{inputenc}
\usepackage{amssymb}
\usepackage{amsmath}
\usepackage{bm}
\usepackage{hyperref}
\usepackage{tikz}
\usepackage{tkz-euclide}
\usepackage{chngcntr}
\usepackage{verbatim}
\usepackage{mathtools}
\usepackage{esvect}
\usepackage{subfiles}
\usepackage{color}
\usepackage{xspace}
\usepackage{xparse}
\usepackage[noend]{algpseudocode}
\usepackage{subfigure}
\usepackage{bbm}
\usepackage{enumerate}
\usepackage{graphicx}
\usepackage{multirow}

\usepackage{pgfplotstable}
\usepackage{pgfplots}
\pgfplotsset{
	table/search path={plot_figures},
}
\usepackage{tikz}
\usepackage{xr}
\usetikzlibrary{external}
\usetikzlibrary{arrows,%
	petri,%
	topaths}%
\usetikzlibrary{graphs,graphs.standard}
\usepackage{tkz-berge}

\usepackage[percent]{overpic}

\usepackage{epsfig} 
\newtheorem{assumption}{Assumption}

\newtheorem{theorem}{Theorem}
\newtheorem{lemma}[theorem]{Lemma}

\newtheorem{corollary}[theorem]{Corollary}

\usepackage{pgfplotstable}
\usepackage{pgfplots}
\usepackage{cite}
\pgfplotsset{
	table/search path={plot_figures},
}
\usepackage{tikz}
\usepackage{xr}
\usetikzlibrary{arrows,%
	petri,%
	topaths}%
\usetikzlibrary{graphs,graphs.standard}
\usepackage{tkz-berge}
\usepackage{subfigure}

\allowdisplaybreaks

\title{\LARGE \bf
	Non-Bayesian Social Learning with Gaussian Uncertain Models
}

\author{James Z. Hare, C\'esar A. Uribe, Lance Kaplan, and Ali Jadbabaie 
	\thanks{This research was sponsored by the DARPA Lagrange, Vannevar Bush Fellowship, and OSD LUCI programs. 
	}
	\thanks{J.Z.H. and L.K. (\textit{\{james.z.hare.civ, lance.m.kaplan.civ\}@mail.mil}) are with the U.S. Army Research Laboratory, Adelphi, MD. C.A.U and A.J. are with the Laboratory for Information and Decision Systems (LIDS), and the Institute for Data, Systems, and Society (IDSS),
		Massachusetts Institute of Technology, Cambridge, MA
		(\textit{\{cauribe,jadbabai\}@mit.edu}).  }%
}
\begin{document}

\maketitle
	\thispagestyle{empty}
	\pagestyle{empty}

\begin{abstract}
Non-Bayesian social learning theory provides a framework for distributed inference of a group of agents interacting over a social network by sequentially communicating and updating beliefs about the unknown state of the world through likelihood updates from their observations. Typically, likelihood models are assumed known precisely. However, in many situations the models are generated from sparse training data due to lack of data availability, high cost of collection/calibration, limits within the communications network, and/or the high dynamics of the operational environment. Recently, social learning theory was extended to handle those model uncertainties for categorical models. In this paper, we introduce the theory of Gaussian uncertain  models and study the properties of the beliefs generated by the network of agents. We show that even with finite amounts of training data, non-Bayesian social learning can be achieved and all agents in the network will converge to a consensus belief that provably identifies the best estimate for the state of the world given the set of prior information.

\end{abstract}

\section{Introduction}
The setting of non-Bayesian social learning~\cite{JMST2012} often assumes that there is a network of boundedly rational agents who are receiving private observations, communicating, and updating beliefs about the model that best represents the underlying truth. In this framework, the beliefs are assumed to be sufficient statistics for what individuals know about the state of the world, and the update is known to suffer from  \textit{imperfect recall}~\cite{MTJ18}, which significantly simplifies combining the agents beliefs as compared to Bayesian social learning theory, at the expense of ``double counting'' information ~\cite{GK2003,ADLO2011,KT2013,RJM2017}. 

The literature of non-Bayesian social learning theory typically studies various social learning rules that allows the agents to sequentially combine and update their beliefs in a manner that aggregates all of the information available in the network. Much of the learning rules developed consider that each agent combines their neighbors' beliefs using a weighted arithmetic~\cite{RT2010,RMJ2014,LJS2018} or a geometric average~\cite{JMST2012,SJ2013}. Then, the beliefs are updated by scaling the combined beliefs by the likelihood of their new observation given that the particular model is the ground truth. Variations of these learning rules have been proposed to handle fixed and time-varying graphs~\cite{UHLJ2019}, weakly-connected graphs~\cite{SYS2017,SYS2018}, increasing self-confidence~\cite{UJ2019}, compact hypotheses sets~\cite{nedic2017distributed}, and adversarial attacks~\cite{HULJ2019,SV2018}.

There are several variations of these social learning rules proposed in the literature which aim to improve the learning rate of the agents. These include using one-step memory~\cite{NOU2017}, observation reuse~\cite{SV2018,BT2018}, and most recently the min-rule~\cite{MRS2019}. Although the current literature has made significant advances in this problem, they all assume that the statistical models used to evaluate the likelihoods are known precisely. This assumption requires that the agents collect a large set of training exemplars to ensure that the estimated models provide an accurate representation. However, in many situations, the amount of training data available may be limited or too expensive to collect, requiring that the agents incorporate their \emph{uncertainty} into the likelihood models.

Modeling uncertainty has been previously studied in the fields of possibility theory \cite{DP2012}, probability intervals \cite{W1997}, and belief functions \cite{S1976, SK1994} by extending probability theory and expressing the likelihood model parameters within a fixed interval. Other approaches follow a Bayesian framework by modeling the uncertainty in the likelihood model parameters as a second-order probability density function \cite{W1996, J2018}, which is typically a conjugate prior of the underlying statistical model. Then, the \emph{uncertain likelihood} model can be computed as the posterior predictive distribution \cite{R1984}.  

Recently, the concept of \textit{uncertain models} for observations drawn from an unknown multinomial distribution was proposed and was included in a social learning setting\cite{HULJ2019_TSP}, which was  later extended to time-varying directed graphs~\cite{UHLJ2019}. This technique proposed an \textit{uncertain likelihood ratio} as the likelihood model to test the consistency of the prior evidence (training data) with the measurement sequence (testing data) \cite{HULJ2019_TSP,UHLJ2019}. The uncertain likelihood ratio is defined as a standard likelihood ratio test, except as a ratio of the posterior predictive distribution conditioned on the prior evidence to the posterior predictive distribution conditioned on zero prior evidence (or non-informative prior). In this regard, the beliefs generated by the social learning rules are evaluated on their own merit. Applying this approach
in the limiting condition when the amount of prior evidence grows unboundedly, the agents accurately infer the ground truth model and achieve the same result as traditional non-Bayesian social learning theory. Additionally, when the amount of prior evidence is finite, uncertain models generalize the problem allowing for a measure of confidence in the inference results. 

In this work, we expand upon the idea of uncertain models \cite{HULJ2019_TSP} to address the scenario when measurements are real-valued and drawn from a Gaussian distribution. We derive the general Gaussian uncertain model for situations where the mean and precision are unknown and implement this model into a standard non-Bayesian social learning rule. We found that the beliefs of every agent converge to the centralized solution, which is a geometric average of their individual uncertain likelihood ratios. Furthermore, as the agents' amount of prior evidence grows unboundedly, the beliefs with Gaussian uncertain models are the same as traditional non-Bayesian social learning theory. This indicates that Gaussian uncertain models can successfully be used as a general inference test for any amount of prior evidence. 

The remainder of this paper is organized as follows. First, in Section~\ref{sec:pfam} we present the problem, our proposed algorithm, and the main results. Then, we derive the Gaussian uncertain models in Section~\ref{sec:ulm} and the uncertain likelihood update utilized in the social learning rule in Section~\ref{sec:nbsl_gum}. Then, we outline the process of proving the main results, in Section~\ref{sec:proofs}. Finally, we provide a numerical analysis in Section~\ref{sec:sim} to empirically validate our results and conclude the paper in Section~\ref{sec:con}.  

\textbf{Notation:} Bold symbols represent a vector/matrix, while a non-bold symbol represents its element. The indexes $i$ and $j$ represent agents and $t$ represents time. We abbreviate the terminology independent identically distributed as i.i.d.. We use $[\mathbf{A}]_{ij}$ to represent the entry of matrix $\mathbf{A}'s$ $i$th row and $j$th column. The empty set is denoted as $\emptyset$. The Gaussian distribution is 
\begin{eqnarray}
\mathcal{N}(\omega|\mu,\lambda^{-1}) = \frac{\sqrt{\lambda}}{\sqrt{2\pi}}e^{-\frac{\lambda(\omega-\mu)^2}{2}},
\end{eqnarray}
and the Gaussian-gamma distribution is
\begin{align}
\mathcal{N}\mathcal{G}(x,\lambda|\mu,\kappa,\alpha,\beta) = \frac{\beta^\alpha \sqrt{\kappa}}{\Gamma(\alpha) \sqrt{2\pi}} \lambda^{\alpha-\frac{1}{2}} e^{-\beta\lambda}e^{-\frac{\kappa\lambda(x-\mu)^2}{2}}.
\end{align}
We denote the Kullback-Liebler (KL) divergence as 
\begin{eqnarray}
D_{KL}(p(x)\|q(x)) = -\int p(x)\log\left(\frac{p(x)}{q(x)} \right)dx,
\end{eqnarray}
where $p(x)$ and $q(x)$ are two continuous probability distributions over $x\in \mathbb{R}$. 

\section{Problem Formulation, Social Learning Rule, and Main Result} \label{sec:pfam}
\subsection{Hypotheses, signals, and prior evidence}
Consider a network of $m$ agents connected over a social network who are trying to infer the unknown state of the world $\theta^*$ from a finite set of states (or hypotheses) $\boldsymbol{\Theta}=\{\theta_1,...,\theta_m\}$. At each time step $t\ge 1$, we assume that each agent $i$ collects an i.i.d. private observation $\omega_{it}\in\mathbb{R}$ sampled from an unknown Gaussian distribution $P_{i\theta^*}=\mathcal{N}(\cdot|\mu_{i\theta^*},\lambda_{i\theta^*}^{-1})$ with mean $\mu_{i\theta^*}$ and precision $\lambda_{i\theta^*} = {1}/{\sigma_{i\theta^*}^2}$\footnote{It is possible to generalize this to within $\mathbb{R}^n$, however, this condition is out of the scope of this paper and will be considered as future work}. We denote the set of measurements received up to time $t$ as $\boldsymbol{\Omega}_{i1:t}$, where the measurements are independent across the agents.\footnote{In general, each agent may have a different measurement model or sensing capability from one another. This can result in $\mu_{i\theta^*}$ and $\lambda_{i\theta^*}$ varying between agents.} The overall goal of the agents is to collectively agree on the hypothesis that best matches the ground truth distribution. 

\begin{figure}
    \centering
    \includegraphics[width=.5\columnwidth]{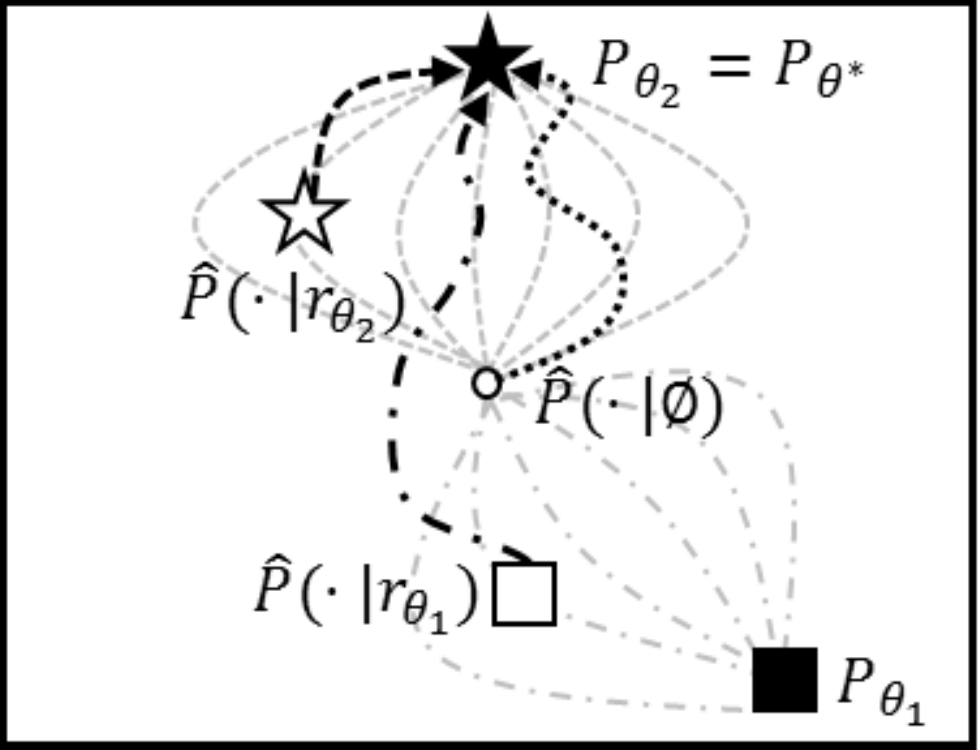}
    \caption{\textbf{Geometric interpretation of uncertain models:} The outer square represents a continuous probability space of the distributions of signals $\omega_{it}$. The solid square and star represent the true distribution for hypothesis $\theta_1$ and $\theta_2$ respectively, while the open square and star are the uncertain distributions for $\theta_1$ and $\theta_2$, where $\theta_2$ is $\theta^*$. The open circle in the center represents the model of complete ignorance. The uncertain distribution is a mixture between $P(\cdot|\emptyset)$ and $P_\theta$, which depends on the amount of prior evidence collected. Zero prior evidence causes the location of $\hat{P}(\cdot|\mathbf{r}_\theta)$ to be $\hat{P}(\cdot|\emptyset)$, while an infinite amount of prior evidence causes it to be located at $P_\theta$. Then, as measurements are received, the distributions $\hat{P}(\cdot|r_{\theta_1}), \hat{P}(\cdot|r_{\theta_2})$, and $\hat{P}(\cdot|\emptyset)$ traverse through the probability space until they eventually collapse on $P_{\theta^*}$ with an infinite amount of measurements. We define the uncertain likelihood ratio $\Lambda_{\theta}=\hat{P}(\cdot|r_\theta)/\hat{P}(\cdot|\emptyset)$ as our consistency test, where a hypothesis with prior evidence consistent with the ground truth will have a shorter trajectory from $\hat{P}(\cdot|r_\theta)$ to $P_{\theta^*}$ than $\hat{P}(\cdot|\emptyset)$ to $P_{\theta^*}$.}
    \label{fig:simplex} \vspace{-15pt}
\end{figure}

Traditionally, each agent undergoes a training phase where they collect a sufficiently large amount of labeled training data to accurately estimate the parameters $\mu_{i\theta}$ and $\lambda_{i\theta}$ of the distribution $P_{i\theta}=\mathcal{N}(\cdot|\mu_{i\theta},\lambda_{i\theta}^{-1})$ for each hypothesis $\theta$. This results in a precisely known statistical model for each $\theta$. However, this work considers that the agents collect a varying amount of prior evidence (training data) for each hypothesis, which may lead to inaccurate estimates of the parameters, requiring uncertain statistical models. 

Consider that an agent $i$ has access to a hypothesis $\theta\in\boldsymbol{\Theta}$ and collects $R_{i\theta}\ge0$ drawn from the distribution $P_{i\theta}=\mathcal{N}(\cdot|\mu_{i\theta},\lambda_{i\theta}^{-1})$, where $\mu_{i\theta}$ and $\lambda_{i\theta}$ are unknown. This results in a set of $R_{i\theta}$ training samples $\mathbf{r}_{i\theta}=\{r_{ik\theta}\}_{\forall k \in\{1,...,R_{i\theta}\}}$ which are then used to estimate the parameter $\mu_{i\theta}$ and $\lambda_{i\theta}^{-1}$. 

Instead of adopting a frequentist's interpretation and simply estimating the parameters by the sample mean and variance, this work implements a Bayesian approach that exploits conjugate distributions to estimate the posterior distribution of $\mu$ and $\lambda$ conditioned on the prior evidence $\mathbf{r}_{i\theta}$. Since the family of statistical models is assumed to be Gaussian with unknown mean and variance, a natural conjugate distribution is the Gaussian-gamma distribution \cite{D2005}. Then, we estimate the \emph{uncertain likelihood} using the parameters posterior distribution by predicting the likelihood of the measurement sequence give the prior evidence, i.e., 
$\hat{P}(\boldsymbol{\Omega}_{i1:t}|\mathbf{r}_{i\theta})$, as the posterior predictive distribution \cite{R1984}. An example of the uncertain likelihood is shown in Fig.~\ref{fig:simplex}.  

Typically in hypothesis testing \cite{D2005}, the likelihoods are normalized over the set of hypotheses and the hypothesis with the maximum likelihood is selected as the ground truth. This can also be thought of as the likelihood distribution closest to the ground truth in the probability space, see Fig.~\ref{fig:simplex}. However, in the uncertain case, the posterior predictive distribution for each hypothesis is computed with a varying amount of prior evidence making them incommensurable \cite{HULJ2019_TSP}.
Thus, we normalize the uncertain likelihood by another posterior predictive distribution of the measurement sequence, except here we use a noninformative Gaussian-gamma prior \cite{G2006} having zero prior evidence, i.e., $\hat{P}(\boldsymbol{\Omega}_{i1:t}|\mathbf{r}_{i\theta}=\emptyset)$. This \emph{uncertain likelihood ratio} $\Lambda_{i\theta}= \hat{P}(\cdot|\mathbf{r}_{i\theta})/\hat{P}(\cdot|\emptyset)$ acts as our uncertain statistical model and is derived in Section~\ref{sec:ulm}. The uncertain likelihood ratio is a consistency test between the prior evidence and the measurement sequence. It quantifies the amount of evidence to accept or reject the hypothesis that a model $\theta$ is distinguishable from the ground truth $\theta^*$.  The set of hypotheses that are indistinguishable from the ground truth for the $i$-th agent is 
$$
\boldsymbol{\Theta}_i^* = \{\theta | \mathcal{N}(\cdot|\mu_{i,\theta},\lambda^{-1}_{i,\theta}) = \mathcal{N}(\cdot|\mu_{i,\theta^*},\lambda^{-1}_{i,\theta^*}) \ \ \forall \theta \in \boldsymbol{\Theta} \}.
$$

We can visually interpret the uncertain likelihood ratio in Fig.~\ref{fig:simplex}. As an agent collects measurements, the uncertain distributions $\hat{P}(\cdot|\mathbf{r}_{\theta_1})$, $\hat{P}(\cdot|\mathbf{r}_{\theta_2})$, and $\hat{P}(\cdot|\emptyset)$ inch their way closer to $P_{\theta^*}$, where their rate depends on the amount of prior evidence collected. The number of time steps that the uncertain distribution is closer/further to $P_{\theta^*}$ than $\hat{P}(\cdot|\emptyset)$ governs much greater than $1$ or closer to $0$ the uncertain likelihood ratio will be, respectively. Thus, our consistency test, presented in Section~\ref{sec:ulm}, accepts/rejects hypotheses with shorter/longer trajectories of $\hat{P}(\cdot|\mathbf{r}_\theta)\to P_{\theta^*}$ than $\hat{P}(\cdot|\emptyset)\to P_{\theta^*}$, respectively. 

\subsection{Social Learning Rule}
Next, we propose the distributed inference algorithm for a group of agents interacting over a social network. Initially at time $t=0$, each agent $i$ constructs a belief $\mu_{i0}(\theta)=1$ for each hypothesis $\theta\in\boldsymbol{\Theta}$, where each belief represents an aggregated uncertain likelihood ratio discussed above and presented in Section~\ref{sec:ulm}. Then, for each time step $t\ge1$, each agent sequentially communicates their beliefs to their neighbors, receives a new observation, and updates their beliefs using a social learning rule. 

We assume that the agents interact over a network modeled as an undirected graph $\mathcal{G}=(\mathcal{M},\mathcal{E}),$\footnote{Note that the results herein hold for directed graphs as long as the the graph satisfies Assumption~\ref{assum:graph}.} where $\mathcal{M}=\{1,...,m\}$ is the set of agents and $\mathcal{E}$ is the set of edges between agents. If agents $i$ and $j$ can communicate their beliefs to each other, then $(i,j)\in\mathcal{E}$. We denote agent $i$'s set of neighbors as $\mathcal{M}_{i} = \{j | (j,i)\in\mathcal{E}, \forall j\in\mathcal{M}\}$ and each edge is assumed to be weighted and modeled as an adjacency matrix $\mathbf{A}$, where $[\mathbf{A}]_{ij}>0$ if $(i,j)\in\mathcal{E}$. 

During each time step $t\ge1$, each agent $i$ has access to the information $\psi_{it}(\theta) = \{\omega_{it+1}, \mathbf{r}_{i\theta}, \mu_{it}(\theta), \{\mu_{jt}(\theta)\}_{\forall j\in\mathcal{M}_i}\}$ for each hypothesis $\theta$. Then, agent $i$ updates their belief $\mu_{it+1}(\theta)$ using the following update rule:
\begin{eqnarray}
\mu_{it+1}(\theta) = \ell_{i\theta}(\omega_{it+1}|\boldsymbol{\Omega}_{i1:t})\prod_{j\in\mathcal{M}_i} \mu_{jt}(\theta)^{[\mathbf{A}]_{ij}}, \label{eq:LL_rule} 
\end{eqnarray}
where the product on the right hand side of (\ref{eq:LL_rule}) represents a geometric average of their neighbors beliefs and $\ell_{i\theta}(\omega_{it+1})$ is the Gaussian uncertain likelihood update defined as\footnote{Here, we have simplified the notation and will only provide the conditioned measurements $\boldsymbol{\Omega}_{i1:t}$ when necessary.}

\begin{align} \label{eq:ell}
\ell_{i\theta}(\omega_{it+1}) =& \frac{\Gamma(\alpha_{R_{i\theta}+t+1})\Gamma(\alpha_{t})\beta_{t+1}^{\alpha_{t+1}}\beta_{R_{i\theta}+t}^{\alpha_{R_{i\theta}+t}}}{\Gamma(\alpha_{t+1})\Gamma(\alpha_{R_{i\theta}+t}) \beta_{t}^{\alpha_{t}} \beta_{R_{i\theta}+t+1}^{\alpha_{R_{i\theta}+t+1}}} \nonumber \\ &\cdot\frac{(\kappa_{t+1}\kappa_{R_{i\theta}+t})^{1/2}}{(\kappa_{t}\kappa_{R_{i\theta}+t+1})^{1/2}}. 
\end{align} 

For simplicity of presentation, we postpone the explicit definition of the uncertain likelihood update parameters to Sections~\ref{sec:ulm} and ~\ref{sec:nbsl_gum} in (\ref{eq:ul_params}), (\ref{eq:ell_params}), and (\ref{eq:ulr_params}). 
Note that $\kappa$ and $\alpha$ represent a count of the data items, while $\beta$ is a centralized sum of squares for the set of these data items. This is a closed form expression of a ratio of predictive posterior distributions, where the numerator is the expected value of  $\mathcal{N}(\omega_{it+1}|\mu,\lambda^{-1})$ taken over the Gaussian-gamma distribution conditioned on the prior evidence $\mathbf{r}_{i\theta}$ and the measurement sequence $\boldsymbol{\Omega}_{i1:t}$ and the denominator is the expected value of $\mathcal{N}(\omega_{it+1}|\mu,\lambda^{-1})$ taken over the Gaussian-gamma distribution conditioned on only the measurement sequence $\boldsymbol{\Omega}_{i1:t}$. 

This function is designed such that the product of Gaussian uncertain likelihood updates $\prod_{\tau=1}^{t+1} \ell_{i\theta}(\omega_{i\tau})$ is equal to the uncertain likelihood ratio at time $t+1$. Therefore, the beliefs $\mu_{it+1}(\theta)$ represent an aggregated geometric average of all of the agents individual uncertain likelihood ratios. 

Next, we provide some assumptions that allow us to quantify where the beliefs generated by the update rule~(\ref{eq:LL_rule}) converge asymptotically with Gaussian uncertain models. 

\begin{assumption}\label{assum:graph}

               The graph $\mathcal{G}$ and matrix $A$ are such that:

               \begin{enumerate}[(a)]

                              \item $A$ is doubly-stochastic with $\left[A\right]_{ij} = a_{ij} > 0$ for  $i\ne j$ if and only if $(i,j)\in E$.

                              \item $A$ has positive diagonal entries, $a_{ii}>0$ for all $i \in V $.

                              \item The graph $\mathcal{G}$ is  connected.

               \end{enumerate}

\end{assumption}

Assumption~\ref{assum:graph} states that the adjacency matrix is ergordic, i.e., aperiodic and irreducible, and is a common assumption in the literature \cite{NOU2017}. This allows every agent to communicate their beliefs throughout the entire network. 

\begin{assumption} \label{as:dis}
There is at least one agent that can distinguish any  $\theta  \ne \theta^*$ so that $\cap_{i\in\mathcal{M}} \boldsymbol{\Theta}_i^* = \{\theta^*\}$.
\end{assumption}
Assumption~\ref{as:dis} guarantees that the collective group of agents can determine $\theta^*$. As a consequence of Theorem~\ref{cor:dogmatic} below, these agents can determine $\theta^*$ with infinite prior evidence (ie., precise models) as $t \rightarrow \infty$. 

\subsection{Main Results}
Now we are ready to present the properties of the beliefs generated using the update rule given by~\eqref{eq:LL_rule}.

\begin{theorem} \label{th:LL}
Let Assumption \ref{assum:graph} hold. Then, the beliefs generated using the update rule (\ref{eq:LL_rule}) have the following property:
\begin{equation}
\lim_{t\to\infty} \mu_{it}(\theta) = \left(\prod_{j=1}^m \widetilde{\Lambda}_{j\theta}\right)^{\frac{1}{m}} 
\end{equation}
for all $i\in\mathcal{M}$ with probability 1 where 
\begin{eqnarray}\label{eq:tilde_ULR}
\widetilde{\Lambda}_{j\theta} = \frac{\mathcal{N}(\mathbf{r}_{j\theta}|\mu_{j\theta^*},\lambda_{j\theta^*}^{-1})}{P(\mathbf{r}_{j\theta})}, 
\end{eqnarray}
is agent $j$'s asymptotic uncertain likelihood ratio and
\begin{eqnarray} \label{eq:prior_prob}
\hat{P}(\mathbf{r}_{j\theta}) = \frac{\Gamma(\alpha_{R_{j\theta}})\beta_0^{\alpha_0}(2\pi)^{-\frac{R_{j\theta}}{2}}}{\Gamma(\alpha_0)\beta_{R_{j\theta}}^{\alpha_{R_{\theta}}}}\left( \frac{\kappa_0}{\kappa_{R_{j\theta}}}\right)^\frac{1}{2},
\end{eqnarray}
is the posterior predictive distribution of the prior evidence conditioned on a noninformative prior with parameters
\begin{align} \label{eq:ul_params}
& \kappa_{R_{j\theta}}=\kappa_0+R_{j\theta}, \ \ \ \ \ \alpha_{R_{j\theta}} = \alpha_0 + \frac{R_{j\theta}}{2}, \nonumber \\
& \beta_{R_{j\theta}} = \beta_0 + \frac{R_{j\theta}}{2}\left(s_{j\theta}^2 +\frac{\kappa_0  (\bar{r}_{j\theta}-\mu_0)^2}{\kappa_{R_{j\theta}}}\right),
\end{align}
$\bar{r}_{j\theta} = (\sum_{k=1}^{R_{j\theta}} r_{jk\theta})/R_{j\theta}$, $s_{j\theta}^2=(\sum_{k=1}^{R_{j\theta}}(r_{jk\theta}-\bar{r}_{j\theta})^2)/R_{j\theta}$, $\mu_0=0$, $\alpha_0=1$, $\beta_0=1$, and $\kappa_0=1$. 
\end{theorem}

Theorem~\ref{th:LL} states that the beliefs converge to the geometric average of the agents' asymptotic uncertain likelihood ratio, which is the likelihood of the prior evidence conditioned on the true parameters, normalized by the total probability of the prior evidence. When the agents have a finite amount of prior evidence, the beliefs converge to a finite value, i.e. $\mu_{it}(\theta)\in(0,\infty)$, where a value much greater than $1$ indicates that the prior evidence is consistent with the ground truth. However, when the amount of prior evidence grows unboundedly, the agents beliefs have the following properties. 

\begin{theorem} \label{cor:dogmatic}
Let Assumption~\ref{assum:graph} 
hold and every agents' amount of prior evidence grows unboundedly. Then, the belief generated using the update rule (\ref{eq:LL_rule}) have the following properties:
\begin{align} 
&\lim_{t\to\infty, R_{i\theta}\to\infty} \mu_{it}(\theta) \rightarrow \infty, \ \text{if $\theta\in\boldsymbol{\Theta}_j^*$ $\forall j \in \mathcal{M}$, and } \nonumber \\
&\lim_{t\to\infty, R_{i\theta}\to\infty} \mu_{it}(\theta) = 0, \ \text{if $\exists j\in \mathcal{M}$ s.t. $\theta\notin\boldsymbol{\Theta}_j^*$},
\end{align}
with probability 1 and in probability respectively. 
\end{theorem}

Given Assumption~\ref{as:dis}, when the agents' likelihood models use an infinite amount of prior evidence, the agents' beliefs for a hypothesis that is not the ground truth will converge to zero via Theorem~\ref{cor:dogmatic}. Likewise, the belief in the ground truth hypotheses goes diverges to infinity. An outline of the proofs for Theorem~\ref{th:LL} and Theorem~\ref{cor:dogmatic} are presented in Section~\ref{sec:proofs}. 

Given that the above properties hold for update rule (\ref{eq:LL_rule}), the agents can use their beliefs to determine if there is sufficient evidence to accept or reject a hypothesis $\theta$. As the amount of prior evidence goes to infinity, only the ground truth hypothesis will be accepted while the others are rejected, which is consistent with traditional non-Bayesian social learning theory.

\section{Gaussian Uncertain Models} \label{sec:ulm}
In this section, we derive the Gaussian uncertain likelihood ratio, discuss the uncertain likelihood ratio test, and present the asymptotic properties of the uncertain likelihood ratio.

\subsection{Uncertain Likelihood Ratio} \label{sec:ulr}
As stated in Section~\ref{sec:pfam}, each agent $i$ has collected a set of prior evidence $\mathbf{r}_{i\theta}$ for each hypothesis $\theta\in\boldsymbol{\Theta}$ to estimate the distribution of $\mu$ and $\lambda$ in the training phase. This is achieved by computing the posterior conjugate distribution of $\mu$ and $\lambda$ conditioned on the prior evidence $\mathbf{r}_{i\theta}$ as follows.
\begin{align}
f(\mu,\lambda|\mathbf{r}_{i\theta}) &= \frac{1}{\hat{P}(\mathbf{r}_{i\theta})}\mathcal{N}(\mathbf{r}_{i\theta}|\mu,\lambda)\mathcal{N}\mathcal{G}\left(\mu,\lambda |\boldsymbol{\phi}_0\right) \nonumber \\
&= \mathcal{N}\mathcal{G}\left(\mu,\lambda |\boldsymbol{\phi}_{R_{i\theta}}\right),
\end{align}
where $\mathcal{N}\mathcal{G}\left(\mu,\lambda |\boldsymbol{\phi}_0\right)$ is a noninformative conjugate prior with parameters $\boldsymbol{\phi}_0 = \{\mu_0,\kappa_0,\alpha_0,\beta_0\}$\footnote{With an abuse of notation, throughout this work we will use $\boldsymbol{\phi}$ to represent the parameters of the Gaussian-gamma distribution.} and $\hat{P}(\mathbf{r}_{i\theta})$ is the total probability of the prior evidence provided in (\ref{eq:prior_prob}); the posterior distribution parameters in $\boldsymbol{\phi}_{R_{i\theta}}=\{\mu_{R_{i\theta}}, \kappa_{R_{i\theta}}, \alpha_{R_{i\theta}}, \beta_{R_{i\theta}}\}$ are $\mu_{R_{i\theta}} = ({\kappa_0 \mu_0 + R_{i\theta} \bar{r}_{i\theta}})/({\kappa_0+R_{i\theta}})$ and (\ref{eq:ul_params}).

The parameters of the prior distribution $\mathcal{N}\mathcal{G}\left(\mu,\lambda |\boldsymbol{\phi}_0\right)$ are ideally chosen to be noninformative. A common approach in the literature is to use Jeffreys prior \cite{J1998}, which suggests to set $\mu_0=0$, $\kappa_0 = 0$, $\alpha_0 = 0$, $\beta_0=0$ to assign a uniform distribution over the parameter space. However, this would lead to an improper posterior conjugate prior and cannot be chosen. In this work, we chose to utilize $\mu_0 = 0$, $\kappa_0=1$, $\alpha_0 = 1$, and $\beta_0 = 1$ based on an empirical analysis that found that smaller values of $\kappa_0$, $\alpha_0$, and $\beta_0$ cause the beliefs for hypothesis $\theta\ne\theta^*$ at time $t=1$ to jump to a value $\gg1$, requiring a larger amount of prior evidence to reject the hypothesis. A detailed analysis of the parameter effects on the overall inference will be studied as a future work. 

Next, the agent collects a sequence of measurements $\boldsymbol{\Omega}_{1:t}=\{\omega_1,...,\omega_t\}$ in the testing phase and computes the uncertain likelihood. Following the derivation in \cite{M2007}, the uncertain likelihood is modeled as the posterior predictive distribution of the measurement sequence conditioned on the prior evidence, 
\begin{align} \label{eq:UL}
 \hat{P}(\boldsymbol{\Omega}_{i1:t}|\mathbf{r}_{i\theta}) &= \int_0^\infty \int_{\mathbb{R}}  \mathcal{N}(\boldsymbol{\Omega}_{i1:t}|\mu,\lambda)f(\mu,\lambda|\mathbf{r}_{i\theta})d\mu d\lambda \nonumber \\
 &=\frac{\Gamma(\alpha_{R_{i\theta}+t}) \beta_{R_{i\theta}}^{\alpha_{R_{i\theta}}}(2\pi)^{-t/2}\kappa_{R_{i\theta}}^{1/2}}{\Gamma(\alpha_{R_{i\theta}})\beta_{R_{i\theta}+t}^{\alpha_{R_{i\theta}+t}}\kappa_{R_{i\theta}+t}^{1/2}},
\end{align}
where the prior parameters are provided in (\ref{eq:ul_params}) and
\begin{align} \label{eq:ell_params}
    & \mu_{R_{i\theta}+t} = \frac{\kappa_{R_{i\theta}}\mu_{R_{i\theta}} + t \bar{\omega}_{it}}{\kappa_{R_{i\theta}}+t}, \ \kappa_{R_{i\theta}+t} = \kappa_0 +R_{i\theta}+t \nonumber \\
    & \alpha_{R_{i\theta} + t} = \alpha_0 + \frac{R_{i\theta}+t}{2}, \nonumber \\
    & \beta_{R_{i\theta}+t} = \beta_{R_{i\theta}} + \frac{s_{it}-t\bar{\omega}_{it}^2}{2} + \frac{\kappa_{R_{i\theta}}t(\bar{\omega}_{it}-\mu_{R_{i\theta}})^2}{2\kappa_{R_{i\theta}+t}}
\end{align}
with $s_{it} = s_{it-1} + \omega_{it}^2$ and $\bar{\omega}_{it} = (\bar{\omega}_{it-1}(t-1)+\omega_{it})/t$ s.t. $s_{i0}=0$ and $\bar{\omega}_{i0} = 0$. This model can be thought of as the expected value of the likelihood of the measurement sequence $\mathcal{N}(\boldsymbol{\Omega}_{i1:t}|\mu,\lambda)$ taken over the prior distribution $f(\mu,\lambda|\mathbf{r}_{i\theta})$, i.e., $\hat{P}(\boldsymbol{\Omega}_{i1:t}|\mathbf{r}_{i\theta}) = \mathbb{E}_{f(\mu,\lambda|\mathbf{r}_{i\theta})}[\mathcal{N}(\boldsymbol{\Omega}_{i1:t}|\mu,\lambda)]$. When the agent has $R_{i\theta}<\infty$ and the number of observations grows unboundedly, the distribution $\hat{P}(\boldsymbol{\Omega}_{i1:t}|\mathbf{r}_{i\theta})$ eventually becomes $P_{i\theta^*}$ with probability $1$ due to the strong law of large numbers, as seen in Fig.~\ref{fig:simplex}. While when the amount of prior evidence grows unboundedly,  $\hat{P}(\cdot|\mathbf{r}_{i\theta})=\mathcal{N}(\cdot|\mu_{i\theta},\lambda_{i\theta}^{-1})$ with probability $1$  and remains a fixed point in Fig.~\ref{fig:simplex} $\forall t\ge 1$. 

As shown in \cite{HULJ2019_TSP} and stated in Section~\ref{sec:pfam}, hypotheses with varying amounts of prior evidence are incommensurable and must be evaluated on their own merit. Thus, the uncertain likelihood \eqref{eq:UL} is normalized by the model of complete ignorance, i.e., the uncertain likelihood with zero prior evidence, to form the uncertain likelihood ratio, 
\begin{align} \label{eq:ULR}
\Lambda_{i\theta}(t) &= \frac{\hat{P}(\boldsymbol{\Omega}_{i1:t}|\mathbf{r}_{i\theta})}{\hat{P}(\boldsymbol{\Omega}_{i1:t}|\mathbf{r}_{i\theta}=\emptyset)} \nonumber \\
&= \frac{\Gamma(\alpha_0)\Gamma(\alpha_{R_{i\theta}+t}) \beta_t^{\alpha_t}\beta_{R_{i\theta}}^{\alpha_{R_{i\theta}}}}{\Gamma(\alpha_t)\Gamma(\alpha_{R_{i\theta}})\beta_0^{\alpha_0}\beta_{R_{i\theta}+t}^{\alpha_{R_{i\theta}+t}}}\left( \frac{\kappa_{R_{i\theta}}\kappa_t}{\kappa_{R_{i\theta}+t}\kappa_0}\right)^\frac{1}{2}, 
\end{align}
where 
\begin{align} \label{eq:ulr_params}
    & \kappa_t = \kappa_0 + t, \ \ \ \ \ \alpha_t = \alpha_0 + \frac{t}{2}, \nonumber \\
    & \beta_t = \beta_{0} + \frac{s_{it}-t\bar{\omega}_{it}^2}{2} +\frac{\kappa_0t(\bar{\omega}_{it}-\mu_0)^2}{2\kappa_t}.
\end{align}

The model of complete ignorance represents the expected value of $\mathcal{N}(\boldsymbol{\Omega}_{i1:t}|\mu,\lambda)$ taken over a noninformatative Gaussian-gamma distribution, i.e., $\hat{P}(\boldsymbol{\Omega}_{i1:t}|\mathbf{r}_{i\theta}=\emptyset) = \mathbb{E}_{\mathcal{N}\mathcal{G}(\mu,\lambda|\boldsymbol{\phi}_0)}[\mathcal{N}(\boldsymbol{\Omega}_{i1:t}|\mu,\lambda)]$. Just like the uncertain likelihood, $\hat{P}(\boldsymbol{\Omega}_{i1:t}|\mathbf{r}_{i\theta}=\emptyset)$ eventually collapses to $\mathcal{N}(\cdot|\mu_{i\theta^*},\lambda_{i\theta^*}^{-1})$ with probability $1$ as seen in Fig.~\ref{fig:simplex}.

Then, the agent can infer if the measurement sequence is consistent with the prior evidence collected for hypothesis $\theta$ by utilizing an uncertain likelihood ratio test based on the following insights:
	\begin{enumerate}
		\item If $\Lambda_{\theta}(t)$ 
converges to a value above one, there is evidence to accept that $\theta$ is consistent with $\theta^*$.  Higher values indicate more evidence to accept $\theta$ as $\theta^*$.
		\item If $\Lambda_{\theta}(t)$ converges to a value below one, there is evidence to reject that $\theta$ is $\theta^*$. 
		Lower values indicate more evidence to reject $\theta$ as $\theta^*$. 
		\item  If $\Lambda_{\theta}(t)$ converges to a value near one, there is not enough evidence to accept or reject $\theta$ as $\theta^*$.
	\end{enumerate}
	
As a practical matter, one can define a threshold $\upsilon>1$ so that the hypothesis is deemed accepted, rejected or unsure if $\Lambda_\theta(t) \geq \upsilon$, $\Lambda_\theta(t) < 1/\upsilon$ and $1/\upsilon \leq \Lambda_\theta(t) < \upsilon$, respectively. The exact choice of thresholds is application dependent to balance the number of false positives and false negatives. 

\subsection{Properties of the uncertain likelihood ratio}
Next, we provide the properties of the Gaussian uncertain likelihood ratio that are necessary for our main results. 
\begin{lemma} \label{lem:ULR}
The uncertain likelihood ratio (\ref{eq:ULR}) of hypothesis $\theta$ converges to $\widetilde{\Lambda}_{i\theta}$ with probability 1 as $t\to\infty$, where $\widetilde{\Lambda}_{i\theta}$ is the asymptotic uncertain likelihood ratio (\ref{eq:tilde_ULR}).
\end{lemma}

\begin{proof}
First, we note that the denominator in (\ref{eq:ULR}) is actually the total probability of the measurement sequence, i.e.,
\begin{align}
\hat{P}(\boldsymbol{\Omega}_{i1:t}) = \int_0^\infty \int_{\mathbb{R}} \mathcal{N}(\boldsymbol{\Omega}_{i1:t}|\mu,\lambda^{-1})\mathcal{N}\mathcal{G}\left(\mu,\lambda |\boldsymbol{\phi}_0\right) d\mu d\lambda. \nonumber
\end{align}
Then, utilizing Bayes rule, we can express (\ref{eq:ULR}) as
\begin{align}
\Lambda_{i\theta}(t) = \int_0^\infty \int_{\mathbb{R}} \frac{\mathcal{N}\mathcal{G}(\mu,\lambda|\boldsymbol{\Omega}_{i1:t})\mathcal{N}(\mathbf{r}_{i\theta}|\mu,\lambda^{-1})}{\hat{P}(\mathbf{r}_{i\theta})}d\mu d\lambda, \nonumber
\end{align}
where we used the fact that
\begin{align}
\mathcal{N}\mathcal{G}(\mu,\lambda|\boldsymbol{\Omega}_{i1:t}) = \frac{\mathcal{N}(\boldsymbol{\Omega}_{i1:t}|\mu,\lambda^{-1})\mathcal{N}\mathcal{G}\left(\mu,\lambda |\boldsymbol{\phi}_0\right)}{\hat{P}(\boldsymbol{\Omega}_{i1:t})}. \nonumber
\end{align}

Then, as the number of measurements grows unboundedly, the means of $\mathcal{N}\mathcal{G}(\mu,\lambda|\boldsymbol{\Omega}_{i1:t})$ are $\lim_{t\to\infty} \mathbb{E}[\mu] = \mu_{i\theta^*}$, and $\lim_{t\to\infty} \mathbb{E}[\lambda] = \lim_{t\to\infty} (\alpha_{t})/(\beta_{t}) = \lambda_{i\theta^*}$,
while the variances are $\lim_{t\to\infty} var(\mu) = \lim_{t\to\infty} (\beta_t)/(\kappa_t(\alpha_t - 1)) =  0$, and $\lim_{t\to\infty} var(\lambda)= \lim_{t\to\infty} (\alpha_t)/(\beta_t^2) =  0$ with probability $1$ due to the strong law of large numbers. This means that $\mathcal{N}\mathcal{G}(\mu,\lambda|\boldsymbol{\Omega}_{1:t})$ collapses to a Dirac-delta function centered at the means as time goes to infinity, i.e., $\delta(\mu-\mu_{i\theta^*}, \lambda - \lambda_{i\theta^*})$. Thus, $
\lim_{t\to\infty} \Lambda_{i\theta}(t) =  {\mathcal{N}(\mathbf{r}_{i\theta}|\mu_{i\theta^*},\lambda_{i\theta^*}^{-1})}/{\hat{P}(\mathbf{r}_{i\theta})}$
with probability 1. 
\end{proof}

This result then leads to the following corollary when the amount of prior evidence collected grows unboundedly. 

\begin{corollary} \label{cor:ULR_Dog}
When the amount of prior evidence grows unboundedly, the uncertain likelihood ratio (\ref{eq:ULR}) of hypothesis $\theta$ has the following property:
\begin{align}
&\lim_{R_{i\theta} \to\infty} \widetilde{\Lambda}_{i\theta} \to \infty, \ \text{if $\mu_{i\theta} = \mu_{i\theta^*}$ and $\lambda_{i\theta} = \lambda_{i\theta^*}$, and} \nonumber \\
&\lim_{R_{i\theta} \to\infty} \widetilde{\Lambda}_{i\theta} = 0, \ \text{if either $\mu_{i\theta} \ne \mu_{i\theta^*}$ or $\lambda_{i\theta} \ne \lambda_{i\theta^*}$.}
\end{align}
\end{corollary}

\begin{proof}
First, (\ref{eq:ULR}) can be rewritten as 
\begin{align}
\lim_{t\to\infty, R_{i\theta}\to\infty} \Lambda_{i\theta}(t) = \lim_{R_{i\theta}\to\infty} \frac{\mathcal{N}\mathcal{G}(\mu_{i\theta^*},\lambda_{i\theta^*}|\mathbf{r}_{i\theta})}{\mathcal{N}\mathcal{G}(\mu_{i\theta^*},\lambda_{i\theta^*}|\boldsymbol{\phi}_0)}
\end{align}
where the right hand side was achieved by multiplying and dividing (\ref{eq:tilde_ULR}) by $\mathcal{N}\mathcal{G}(\mu_{i\theta^*},\lambda_{i\theta^*}|\boldsymbol{\phi}_0)$ and applying Bayes rule. Then, following the approach in the proof of Lemma~\ref{lem:ULR}, $\lim_{R_{i\theta}\to\infty}\mathcal{N}\mathcal{G}(\mu_{i\theta^*},\lambda_{i\theta^*}|\boldsymbol{\phi}_{R_{i\theta}})\to \delta(\mu_{i\theta} - \mu_{i\theta^*},\lambda_{i\theta} - \lambda_{i\theta^*})$ with probability $1$ due to the strong law of large numbers. Then, since $\mathcal{N}\mathcal{G}(\mu_{i\theta^*},\lambda_{i\theta^*}|\boldsymbol{\phi}_0)$ is a strictly positive distribution $\forall \mu$ and $\lambda$, $\Lambda_{i\theta}(t)$ will diverge to $\infty$ if $\mu_{i\theta}=\mu_{i\theta^*}$ and $\lambda_{i\theta}=\lambda_{i\theta^*}$, or converge to $0$ if either $\mu_{i\theta}\ne\mu_{i\theta^*}$ or $\lambda_{i\theta}\ne\lambda_{i\theta^*}$.
\end{proof}

Lemma~\ref{lem:ULR} and Corollary~\ref{cor:ULR_Dog} provide insights into where an individual agents uncertain likelihood ratio converges, which can be used to design $\upsilon$ in the uncertain likelihood ratio test. 

Furthermore, Corollary~\ref{cor:ULR_Dog} can visually be interpreted in Fig.~\ref{fig:simplex}, where as $R_{i\theta}\to\infty$, the uncertain distributions $\hat{P}(\cdot|\mathbf{r}_{\theta_1})$ and $\hat{P}(\cdot|\mathbf{r}_{\theta_2})$ are fixed points located at the solid shapes and $\hat{P}(\cdot|\emptyset)$ continues to follow its trajectory. For $\theta_2$ the expected $\Lambda_{i\theta}(t)$ will be greater than $1$ for all $t$ since it is always closer to $P_{\theta^*}$ than $\hat{P}(\cdot|\emptyset)$, causing it to diverge to $\infty$. Whereas for $\theta_1$, there is always going to be a finite time $T$ where $\forall t>T$, $\hat{P}(\cdot|\emptyset)$ is closer to $P_{\theta^*}$ than $P_{\theta_1}$. Thus, the expected $\Lambda_{i\theta}(t)$ will be less than $1$ and will eventually converge to $0$. 

\section{Gaussian Uncertain Likelihood Update} \label{sec:nbsl_gum}
In the previous section, the Gaussian uncertain model was presented where we assumed that an agent $i$ has received the entire measurement sequence up to time $t$, i.e., $\boldsymbol{\Omega}_{i1:t}$. However, in the social setting, each agent receives a new measurement $\omega_{it}$ at each time step $t$, requiring a recursive formulation of the $\Lambda_{i\theta}(t)$ that allows for new information. Particularly, we can express the uncertain likelihood at each time $t$ as follows:
\begin{eqnarray}\label{eq:tell}
\Lambda_{i\theta}(t) = \prod_{\tau=1}^t \frac{\Lambda_{i\theta}(\tau)}{\Lambda_{i\theta}(\tau-1)} = \prod_{\tau=1}^t \ell_{i\theta}(\omega_{i\tau}).
\end{eqnarray}

The uncertain likelihood update $\ell_{i\theta}(\omega_{it})$ ensures that the agents beliefs are an aggregated mixture of each agents $\Lambda_{i\theta}(t)$ $\forall i\in\mathcal{M}$. Next, we discuss the properties of the uncertain likelihood update $\ell_{i\theta}(\omega_{it})$ that enable our main result. 

\begin{lemma} \label{lem:ulru_finite}
The uncertain likelihood update has the following properties with probability 1:
\begin{enumerate}
    \item $\lim_{t\to\infty} \ell_{i\theta}(\omega_{it}) = 1$ when $R_{i\theta}<\infty$, and
    \item $\lim_{t\to\infty, R_{i\theta}\to\infty} \ell_{i\theta}(\omega|\boldsymbol{\Omega}_{i1:t-1}) = \frac{\mathcal{N}(\omega|\mu_{i\theta},\lambda_{i\theta})}{\mathcal{N}(\omega|\mu_{i\theta^*},\lambda_{i\theta^*})}$.
\end{enumerate}
\end{lemma}

\begin{proof}
We first prove condition $1$. Generally, the uncertain likelihood update (\ref{eq:ell}) can be written as follows. 
\begin{align} \label{eq:ell_gen}
\ell_{i\theta}(\omega_{it}) &=  \frac{\int_{0}^\infty \int_{\mathbb{R}} \frac{\mathcal{N}(\omega_{it}|\mu,\lambda)\mathcal{N}(\boldsymbol{\Omega}_{i1:t-1}|\mu,\lambda)\mathcal{N}\mathcal{G}(\mu,\lambda|\boldsymbol{\phi}_{R_{i\theta}})d\mu d\lambda}{\int_{0}^\infty \int_{\mathbb{R}} \mathcal{N}(\boldsymbol{\Omega}_{i1:t-1}|\mu,\lambda)\mathcal{N}\mathcal{G}(\mu,\lambda|\boldsymbol{\phi}_{R_{i\theta}})d\mu d\lambda}} {\int_{0}^\infty \int_{\mathbb{R}} \frac{\mathcal{N}(\omega_{it}|\mu,\lambda)\mathcal{N}(\boldsymbol{\Omega}_{i1:t-1}|\mu,\lambda)\mathcal{N}\mathcal{G}(\mu,\lambda|\boldsymbol{\phi}_0)d\mu d\lambda}{\int_{0}^\infty \int_{\mathbb{R}} \mathcal{N}(\boldsymbol{\Omega}_{i1:t-1}|\mu,\lambda)\mathcal{N}\mathcal{G}(\mu,\lambda|\boldsymbol{\phi}_0)d\mu d\lambda}} \nonumber \\
&= \frac{\int_{0}^\infty \int_{\mathbb{R}} \mathcal{N}(\omega_{it}|\mu,\lambda) \mathcal{N}\mathcal{G}(\mu,\lambda| \boldsymbol{\phi}_{R_{i\theta} + t-1})d\mu d\lambda} {\int_{0}^\infty \int_{\mathbb{R}} \mathcal{N}(\omega_{it}|\mu,\lambda) \mathcal{N}\mathcal{G}(\mu,\lambda| \boldsymbol{\phi}_{t-1})d\mu d\lambda},\nonumber \\
\end{align}
where the first line is achieved due to i.i.d. measurements, while the second line is an application of Bayes' rule. As illustrated in the proof of Lemma~\ref{lem:ULR}, as the measurement sequence grows unboundedly,  $\lim_{t\to\infty} \mathcal{N}\mathcal{G}(\mu,\lambda|\boldsymbol{\phi}_{t-1})=\delta(\mu-\mu_{i\theta^*},\lambda-\lambda_{i\theta^*})$ and $\lim_{t\to\infty} \mathcal{N}\mathcal{G}(\mu,\lambda|\boldsymbol{\phi}_{R_{i\theta} + t-1})=\delta(\mu-\mu_{i\theta^*},\lambda-\lambda_{i\theta^*})$ with probability 1 since $R_{i\theta}<\infty$. Thus, 
\begin{eqnarray}
\lim_{t\to\infty}\ell_{i\theta}(\omega|\boldsymbol{\Omega}_{i1:t-1}) = \frac{\mathcal{N}(\omega|\mu_{i\theta^*},\lambda_{i\theta^*})}{\mathcal{N}(\omega|\mu_{i\theta^*},\lambda_{i\theta^*})} =  1.  \nonumber
\end{eqnarray}

Next, we prove condition $2$ when the amount of prior evidence grows unboundedly. Following the same logic as above, $\lim_{R_{i\theta}\to\infty} \mathcal{N}\mathcal{G}(\mu,\lambda|\mathbf{r}_{i\theta}) = \delta(\mu-\mu_{i\theta}, \lambda-\lambda_{i\theta})$ with probability $1$. Then, the $\ell_{i\theta}(\omega)$ simplifies to
\begin{align}\label{eq:ell_rinfity}
\lim_{R_{i\theta}\to\infty} \ell_{i\theta}(\omega_{it}) = \frac{\mathcal{N}(\omega_{it}|\mu_{i\theta},\lambda_{i\theta})}{\int_{0}^\infty \int_{\mathbb{R}} \mathcal{N}(\omega_{it}|\mu,\lambda) \mathcal{N}\mathcal{G}(\mu,\lambda| \boldsymbol{\phi}_{t-1})d\mu d\lambda}. 
\end{align}
Thus, as the number of private signals grows unboundedly, the uncertain likelihood update converges to 
\begin{align} \label{eq:ell_trinfty}
\lim_{t\to\infty,R_{i\theta}\to\infty} \ell_{i\theta}(\omega|\boldsymbol{\Omega}_{it-1}) = \frac{\mathcal{N}(\omega|\mu_{i\theta},\lambda_{i\theta})}{\mathcal{N}(\omega|\mu_{i\theta^*},\lambda_{i\theta^*})}, 
\end{align}
with probability $1$ for any $\omega\in\mathbb{R}$. 
\end{proof}

\begin{corollary} \label{cor:ell_matching}
When $R_{i\theta}\to\infty$ and $\mathbf{r}_{i\theta}$ is drawn from the ground truth distribution, i.e., $\mu_{i\theta}=\mu_{i\theta^*}$ and $\lambda_{i\theta}=\lambda_{i\theta^*}$, then the uncertain likelihood update converges to $1$ with probability $1$ as $t\to\infty$. 
\end{corollary}

These properties are critical in proving that the beliefs of every agent converge (or diverge). When the agents have a finite amount of prior evidence, the combined beliefs are updated using $\ell_{i\theta}(\omega_{it})=1$, which turns the social learning rule (\ref{eq:LL_rule}) into a consensus geometric average. Whereas, when the agents prior evidence grows unboundedly, we can express the beliefs as a function of the expected value of the log-uncertain likelihood update captured in the following lemma.  

\begin{lemma} \label{cor:exp_log_ell}
The expected value of the log-uncertain likelihood update when the agent $i$'s amount of prior evidence grows unboundedly has the following properties,
\begin{align}\label{eq:exp_bound_ell}
\mathbb{E}[\log(\ell_{i\theta}(\omega_{it}))] = D_{KL}\left(\mathcal{N}(\cdot|\mu_{i\theta^*},\lambda_{i\theta^*}^{-1})\|\hat{P}(\cdot|\boldsymbol{\Omega}_{i1:t-1})\right)- \nonumber \\ D_{KL}\left(\mathcal{N}(\cdot|\mu_{i\theta^*},\lambda_{i\theta^*}^{-1})\|\mathcal{N}(\cdot|\mu_{i\theta},\lambda_{i\theta}^{-1}))\right), 
\end{align}
where
\begin{eqnarray}
\hat{P}(\cdot|\boldsymbol{\Omega}_{i1:t-1}) =\int_0^\infty \int_{\mathbb{R}} \mathcal{N}(\omega|\mu,\lambda)\mathcal{N}\mathcal{G}(\mu,\lambda|\boldsymbol{\phi}_0)d\mu d\lambda
\end{eqnarray}
is a student-t distribution \cite{M2007} and 
{\small
\begin{align}
\lim_{t\to\infty} \mathbb{E}[\log(\ell_{i\theta}(\omega_{it}))] = -D_{KL}(\mathcal{N}(\cdot|\mu_{i\theta^*},\lambda_{i\theta^*}^{-1})\|\mathcal{N}(\cdot|\mu_{i\theta},\lambda_{i\theta}^{-1}))). 
\end{align}}
\end{lemma}

\begin{proof}
First, the proof of Lemma~\ref{lem:ulru_finite} showed that $R_{i\theta}\to\infty$, the uncertain likelihood update is 
\begin{eqnarray}
\ell_{i\theta}(\omega_{it}) = \frac{\mathcal{N}(\omega_{it}|\mu_{i\theta},\lambda_{i\theta}^{-1})}{\hat{P}(\omega_{it}|\boldsymbol{\Omega}_{i1:t-1})}, \nonumber
\end{eqnarray}
with probability 1. Then, the expected value of the log-uncertain likelihood update is
\begin{align}
 &\mathbb{E}[\log(\ell_{i\theta}(\omega_{it}))] = \nonumber \\ &\int_{\mathbb{R}} \mathcal{N}(\omega|\mu_{i\theta^*},\lambda_{i\theta}^{-1})\log\left(\frac{\mathcal{N}(\omega|\mu_{i\theta},\lambda_{i\theta}^{-1})}{\hat{P}(\omega|\boldsymbol{\Omega}_{i1:t-1})}\right)d\omega. 
\end{align}
After adding and subtracting $\mathcal{N}(\omega|\mu_{i\theta^*},\lambda_{i\theta}^{-1})\log(\mathcal{N}(\omega|\mu_{i\theta^*},\lambda_{i\theta}^{-1}))$ inside the integral, we achieve
\begin{multline}
\mathbb{E}[\log(\ell_{i\theta}(\omega))] =  D_{KL}(\mathcal{N}(\cdot|\mu_{j\theta^*},\lambda_{j\theta^*}^{-1})\|\hat{P}(\cdot|\boldsymbol{\Omega}_{i1:t-1}))- \nonumber \\
 D_{KL}(\mathcal{N}(\cdot|\mu_{j\theta^*},\lambda_{j\theta^*}^{-1})\|\mathcal{N}(\cdot|\mu_{j\theta},\lambda_{j\theta}^{-1}))).
\end{multline}
Then, as $t\to\infty$, $\hat{P}(\cdot|\boldsymbol{\Omega}_{i1:t-1})$ converges to a Guassian distribution $\mathcal{N}(\cdot|\mu_{i\theta^*},\lambda_{i\theta^*}^{-1})$ with probability $1$ due to the strong law of large numbers. Thus, our desired result is achieved since
\begin{eqnarray}
\lim_{t\to\infty} D_{KL}(\mathcal{N}(\cdot|\mu_{j\theta^*},\lambda_{j\theta^*}^{-1})\|\hat{P}(\cdot|\boldsymbol{\Omega}_{i1:t-1})) = 0.
\end{eqnarray}
\end{proof}

Lemma~\ref{cor:exp_log_ell} indicates that as time $t$ becomes very large, $\ell_{i\theta}(\omega)$ behaves as $\exp(-D_{KL}(\mathcal{N}(\cdot|\mu_{i\theta^*},\lambda_{i\theta^*}^{-1})\|\mathcal{N}(\cdot|\mu_{i\theta},\lambda_{i\theta}^{-1})) + \epsilon)$ for some $\epsilon>0$, where $\epsilon\to0$ as $t\to\infty$. This means that if $D_{KL}(\mathcal{N}(\cdot|\mu_{i\theta^*},\lambda_{i\theta^*}^{-1})\|\mathcal{N}(\cdot|\mu_{i\theta},\lambda_{i\theta}^{-1})) >\epsilon$, then the expected beliefs will decrease exponentially based on the KL divergence. This result is necessary to prove Theorem~\ref{cor:dogmatic}. 

Finally, we provide the final property of the uncertain likelihood update that is necessary to prove our main result. 
\begin{lemma} \label{lem:ell_bounded}
The uncertain likelihood update is finite and lower bounded by a positive value with probability 1, i.e., $\ell_{i\theta}(\omega_{it})>0$ and finite $\forall t$ with probability 1, for any $t\geq 0$, and any realization $\omega_{it}$ and $i\in \mathcal{M}$. 
\end{lemma}
\begin{proof}
First, for a finite $t$ and $R_{i\theta}$, we note that $\ell_{i\theta}(\omega_{it})$ \eqref{eq:ell_gen} is a ratio of posterior predictive distribution, which are continuous functions, strictly positive $\forall \omega_{it}\in\mathbb{R}$, and proper. Then, when $R_{i\theta}\to\infty$ and $t$ is finite, $\ell_{i\theta}(\omega_{it})$ becomes \eqref{eq:ell_rinfity}, which has the same properties since the numerator is a Gaussian distribution. Furthermore, in the limiting condition when both $t\to\infty$ and $R_{i\theta}\to\infty$, $\ell_{i\theta}(\omega)$ becomes a ratio of Gaussian distributions with the same properties \eqref{eq:ell_trinfty}. Thus, in all three scenarios, $\ell_{i\theta}(\omega)$ can never be $0$ or $\infty$ since the distributions are proper and strictly positive. 
\end{proof}

\section{Outline of the proofs of Theorems~\ref{th:LL} and~\ref{cor:dogmatic}} \label{sec:proofs}
In this section, we will outline how to prove the main results. However, we will not explicitly show the details of the proofs due to space requirements. 

\subsection{Sketch of Theorem~\ref{th:LL} Proof}
To prove convergence, we must show a $t\to\infty$ $\|\log(\boldsymbol{\mu}_{t}(\theta))-((\sum_{j=1}^m \log(\widetilde{\Lambda}_{j\theta}))/m)\mathbf{1}\mathbf{1}'\|\to 0$ with probability $1$, where $\boldsymbol{\mu}_{t}(\theta)$ is a vector of the agents beliefs and $\mathbf{1}$ is a vector of all ones. Noting that $\log(\boldsymbol{\mu}_{t}(\theta))=\sum_{\tau=0}^t \mathbf{A}^{t-\tau}\log(\boldsymbol{\ell}_{\theta}(\omega_\tau))$ and using \eqref{eq:tell}, we can bound this absolute difference as $\sum_{\tau=0}^t \|\mathbf{A}^{t-\tau}-(\mathbf{1}\mathbf{1}')/m\| \|\log(\boldsymbol{\ell}_\theta(\boldsymbol{\omega_\tau}))\|$, where $\boldsymbol{\ell}_{i\theta}(\boldsymbol{\omega}_{\tau})$ is a vector of the individual uncertain likelihood updates. Noting that as $t\to\infty$, $\log(\ell_{i\theta}(\omega_t)) \to 0$ and  $\|\mathbf{A}^{t}-(\mathbf{1}\mathbf{1}')/m\|<\sqrt{2}m\lambda^t$, where $\lambda<1$ is the second largest eigenvalue of the adjacency matrix, we can directly use Lemma 3.1 in \cite{RNV2010} to achieve our desired result. Thus, the beliefs converge to the centralized solution. 

\subsection{Sketch of Theorem~\ref{cor:dogmatic} Proof}
Starting with the condition $\theta_i = \theta_i^*$ for all $i\in\mathcal{M}$, we first show that the log-beliefs diverge to infinity following the same logic as in the sketch of Theorem~\ref{th:LL} proof above. Using the fact that the uncertain likelihood ratio diverges to $\Lambda_{i\theta}\to\infty$ according to Corollary~\ref{cor:ULR_Dog}; the uncertain likelihood update converges to $\ell_{i\theta}(\omega)=1$ according to Corollary~\ref{cor:ell_matching} and is finite according to Lemma~\ref{lem:ell_bounded}, we can follow the same process as above to achieve the desired result. 

For the condition $\theta_i\ne\theta_i^*$ for at least one agent $i$, we first expand the log-belief equation $\log(\boldsymbol{\mu}_t(\theta))=\sum_{\tau = 0}^t \mathbf{A}^{t-\tau}\log(\boldsymbol{\ell}_{i\theta}(\boldsymbol{\omega}_t))$ into a sum of three terms, $\sum_{\tau = 0}^{T_1}\mathbf{A}^{t-\tau}\log(\boldsymbol{\ell}_{i\theta}(\boldsymbol{\omega}_{\tau}))$, $\sum_{\tau = T_1+1}^{t-T_2}\mathbf{A}^{t-\tau}\log(\boldsymbol{\ell}_{i\theta}(\boldsymbol{\omega}_{\tau}))$, and $\sum_{\tau = t-T_2}^{t}\mathbf{A}^{t-\tau}\log(\boldsymbol{\ell}_{i\theta}(\boldsymbol{\omega}_{\tau}))$. We know that because $\log(\ell_{i\theta}(\omega_{it}))$ is finite according to Lemma~\ref{lem:ell_bounded}, the first and third terms are finite. Then, we can pick $T_1$ and $T_2$ large enough such that $|\log(\ell_{i\theta}(\omega_{iT_1}))-\mathbb{E}[\log(\ell_{i\theta}(\omega))]|< \epsilon$ and $\|\mathbf{A}^{t-\tau}-(\mathbf{1}\mathbf{1}')/m\|<\epsilon$ for some $\epsilon>0$. Then, using the law of large numbers, we upper bound the second term by $(t-T_1-T_2)(\frac{1}{m}\sum_{i=1}^{m}\mathbb{E}[\log(\ell_{i\theta}(\omega))]+\epsilon B)$ where $B>0$ is finite and $\mathbb{E}[\log(\ell_{i\theta}(\omega))]$ is the negative KL divergence between $\theta$ and $\theta^*$.  Since $\epsilon$ can be made arbitrarily small by picking larger $T$'s, this upper bound goes to $-\infty$ as $t \rightarrow \infty$. Then, since the exponential function is continuous, the beliefs converge to $0$.

\section{Simulation Study} \label{sec:sim}
\begin{figure}[t]
	\subfigure[$\theta_1= \theta^*$]{
		\centering
		\includegraphics[width=0.5\columnwidth]{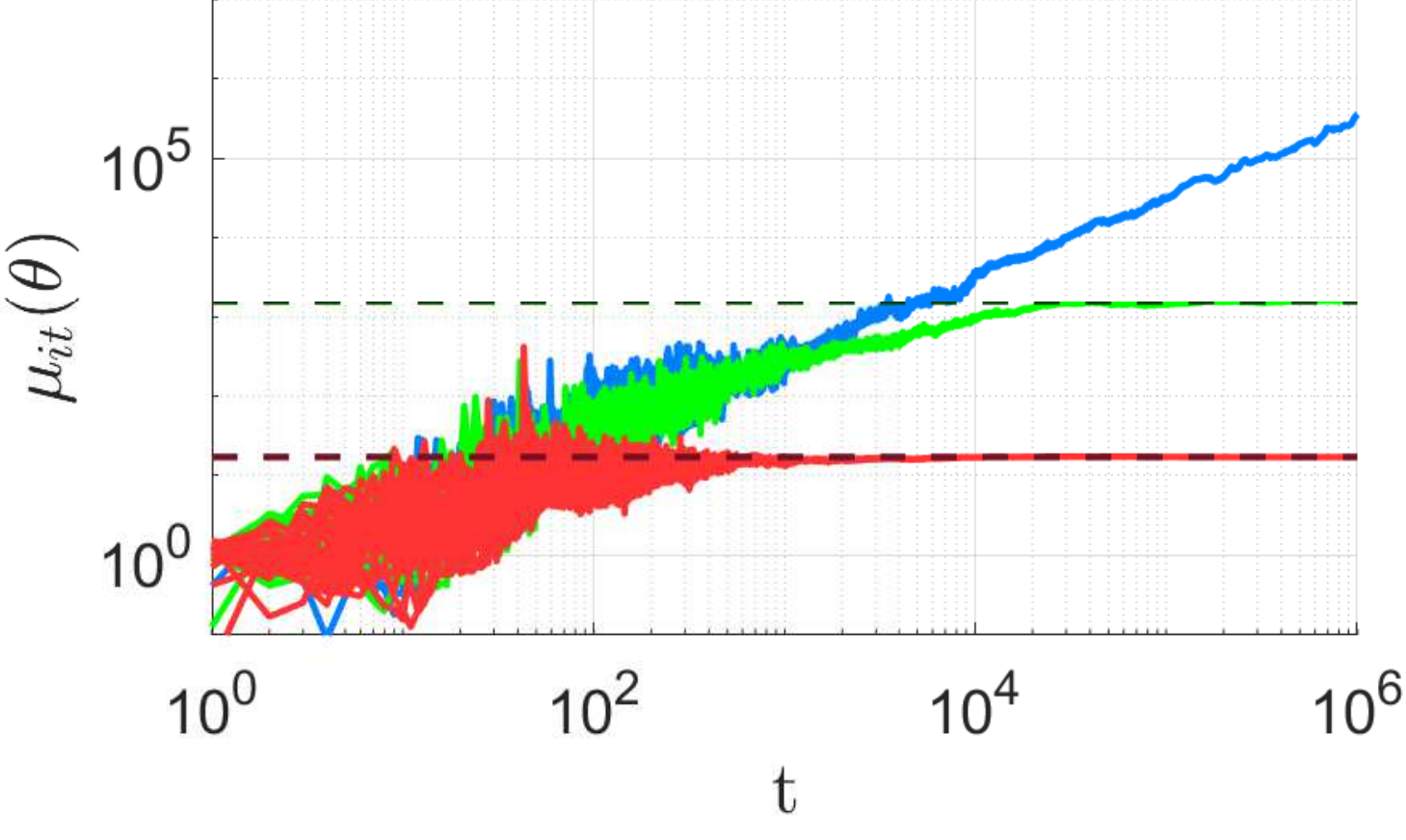}
		\label{fig:t1}
	}%
	\subfigure[$\theta_2\ne \theta^*$]{
		\centering
		\includegraphics[width=0.5\columnwidth]{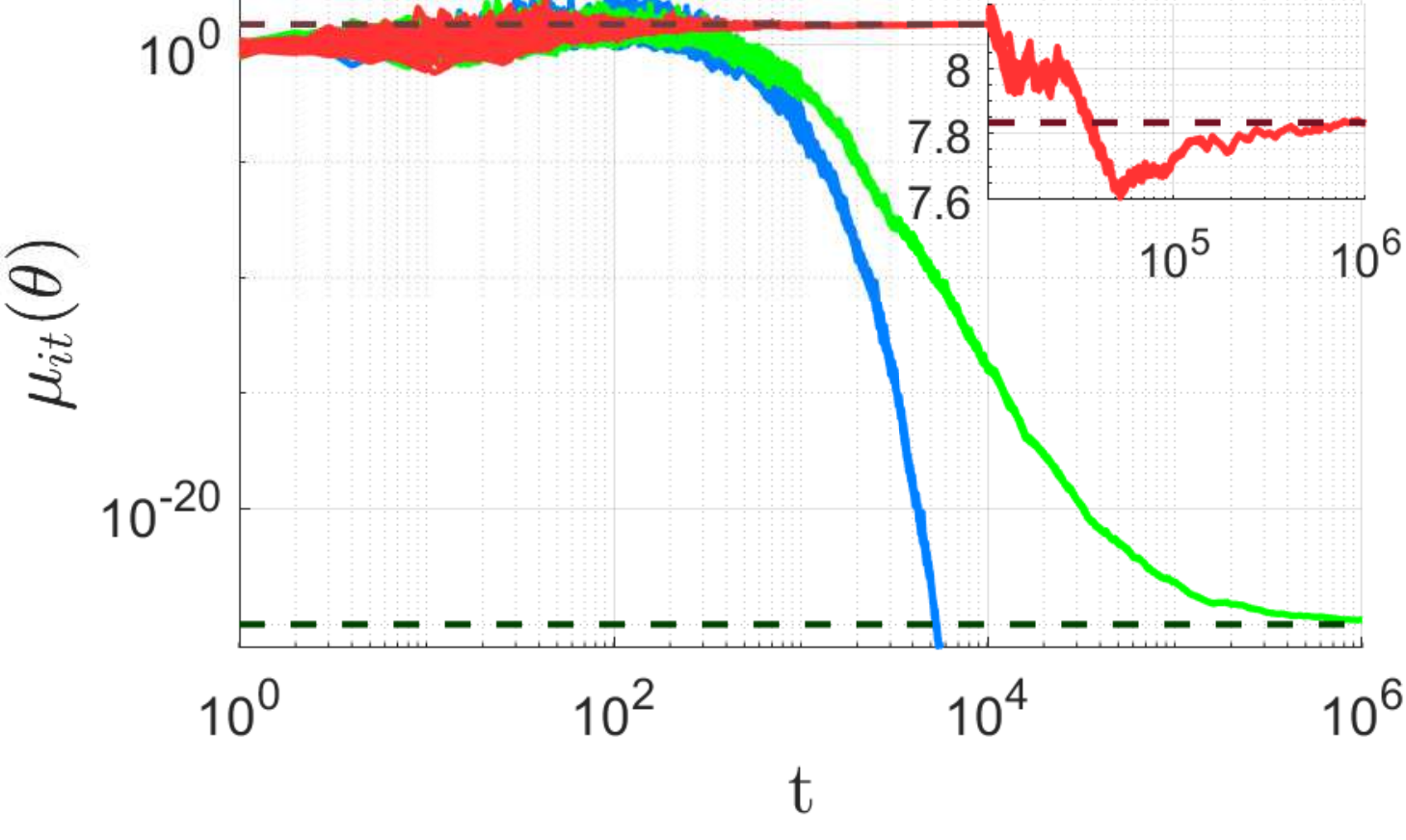}
		\label{fig:t2}
	}%
	
	\vspace{-6pt}
	\centering
	\subfigure{
	    \centering
		\includegraphics[width=.95\columnwidth]{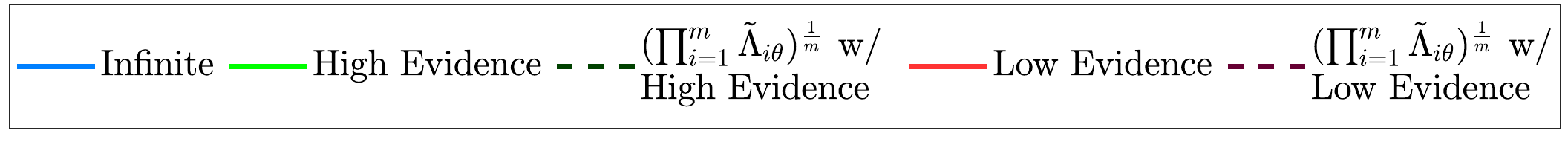}
	} \vspace{-12pt}
	\caption{Evolution of beliefs updated using (\ref{eq:LL_rule}) with $30$ agents connected in a directed cycle graph with self-loops. }\label{fig:mu_graphs} \vspace{-15pt}
\end{figure}

In this section, we empirically validate the convergence properties presented in Theorems~\ref{th:LL} and~\ref{cor:dogmatic}. We simulate a network of $|\mathcal{M}|\in[10,20,30]$ agents connected in an directed cycle graph with self-loops such that the weight on each edge is $0.5$. The agents have a finite set of hypotheses $\boldsymbol{\Theta}=\{\theta_1,\theta_2\}$, where the true parameters for each hypothesis are $\mu_{i\theta_1}=0$, $\lambda_{i\theta_1}=0.5$, $\mu_{i\theta_2}=0$, and $\lambda_{i\theta_2}=0.4$ $\forall i\in\mathcal{M}$ so that $\theta^*=\theta_1$. At each time step $t\ge 1$, each agent receives a measurement drawn from the ground truth distribution with mean $\mu_{i\theta^*}=0$ and precision $\lambda_{i\theta^*}=0.5$ $\forall i\in\mathcal{M}$. In the training phase, the amount of prior evidence collected by each agent is randomly chosen within three categories, Low Evidence with $R_{i\theta}\in[0, 100]$, High Evidence with $R_{i\theta}\in[10^3,10^4]$, and Infinite Evidence where we set $R_{i\theta}$ to a very large number. Then, the network is simulated for $T=10^6$ time steps with the belief update rule \eqref{eq:LL_rule}. 

First, in Fig.~\ref{fig:mu_graphs}, we present the evolution of beliefs for each agent, category of evidence, and hypothesis. Additionally, the dotted lines represent the beliefs point of convergence, i.e., $(\prod_{j=1}^m \widetilde{\Lambda}_{j\theta})^\frac{1}{m}$, present in Theorem~\ref{th:LL}. As seen, the amount of prior evidence directly effects the beliefs point of convergence. When the prior evidence is low, the beliefs converge to a value near $1$ since their initialized uncertain likelihood model is close to the model of complete ignorance. Then, as the amount of evidence increases, the uncertain likelihood model becomes closer to the truth distribution of the hypothesis, causing the beliefs to converge to a larger or smaller value. Furthermore, as the evidence grows unboundedly, the beliefs of $\theta_1$ trend toward $\infty$, while the beliefs of $\theta_2$ converge to $0$, as presented in Theorem~\ref{cor:dogmatic}. 

Fig.~\ref{fig:mu_graphs} also indicates that the beliefs are converging to $(\prod_{j=1}^m \widetilde{\Lambda}_{j\theta})^\frac{1}{m}$. To further validate this result, we simulated the network of agents with a fixed amount of prior evidence within each of the three categories for $50$ Monte Carlo simulation runs, where during each run, a new set of measurements were drawn by each agent. Then, we computed the average log-difference between the beliefs and the centralized solution as seen in Fig.~\ref{fig:con_graphs}. The speed of convergence seems relatively unaffected by the number of agents. On the other hand, as the amount of prior evidence increases, the log-difference is larger at a given value of $t$ because the converged values are larger/smaller for $\theta_1$/$\theta_2$. Also, it takes longer to burn off the effects of the larger prior evidence. Still, the log-difference continues to decay as $t$ increases, indicating convergence. 

\begin{figure}[t]
	\subfigure[$\theta_1= \theta^*$]{
		\centering
		\includegraphics[width=0.5\columnwidth]{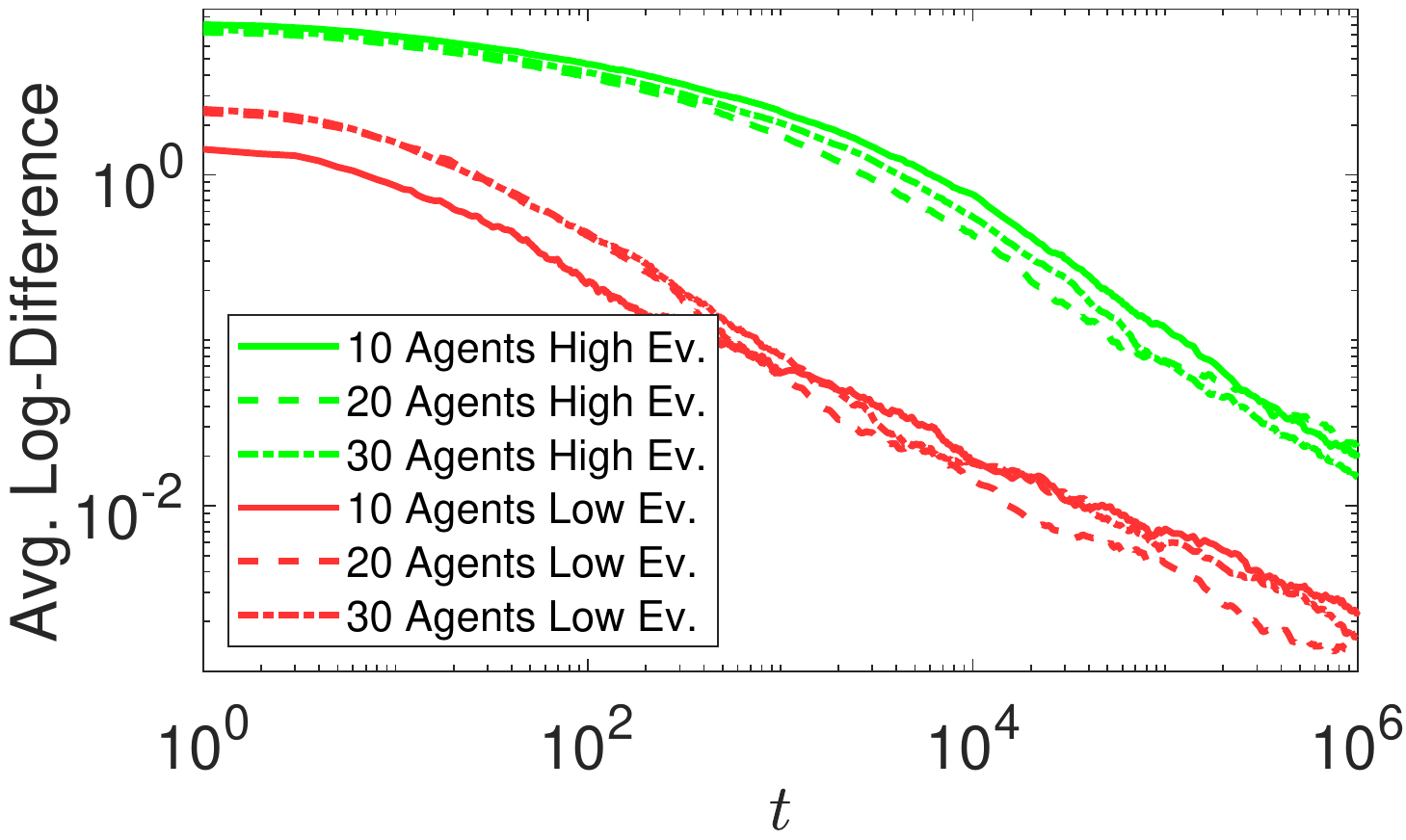}
		\label{fig:con_t1}
	}%
	\subfigure[$\theta_2\ne \theta^*$]{
		\centering
		\includegraphics[width=0.5\columnwidth]{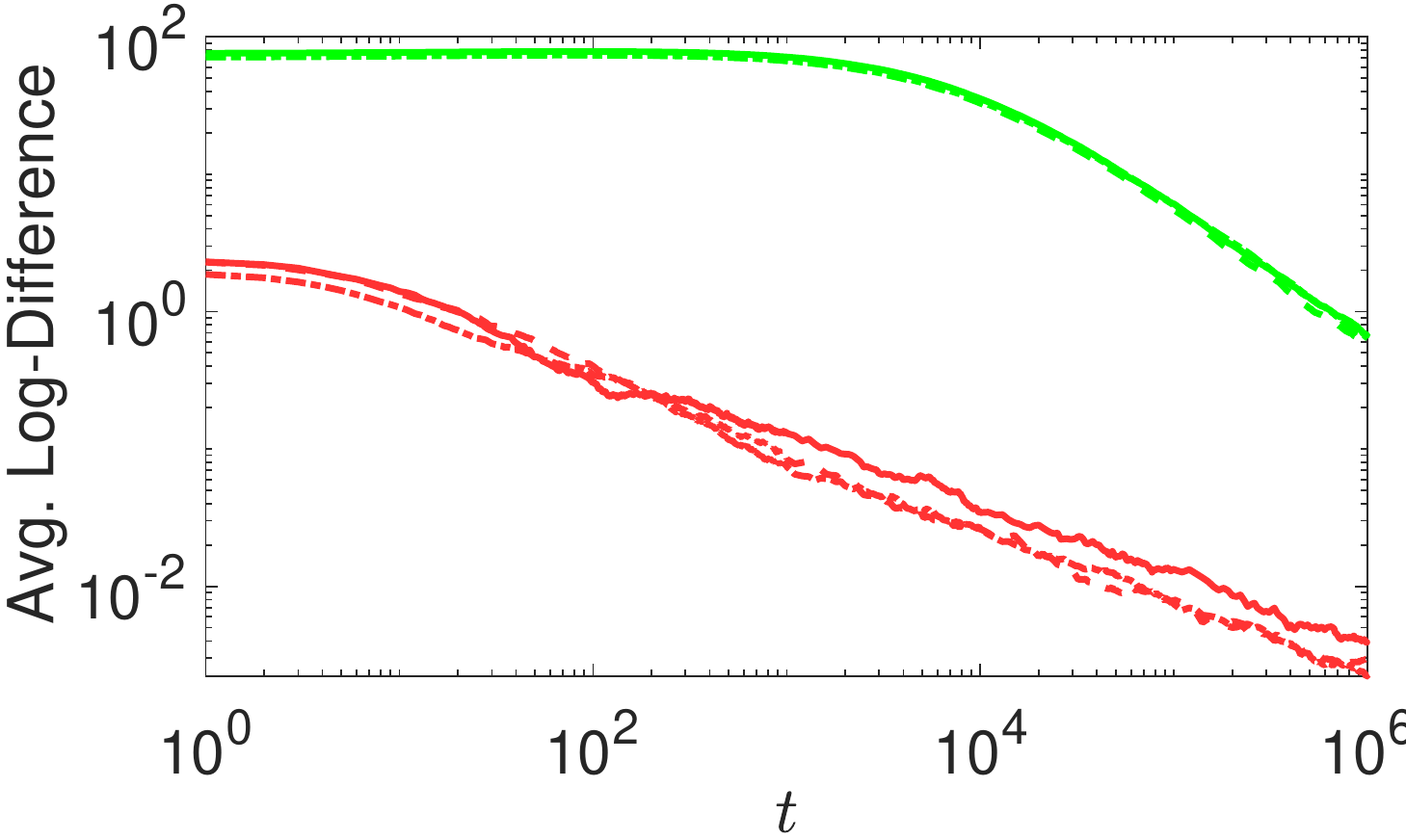}
		\label{fig:con_t2}
	}%

	\caption{The ensemble average difference between the log-beliefs $\log(\mu_{it}(\theta))$ and the log-centralized solution $(\sum_{j=1}^m \log(\widetilde{\Lambda}_{j\theta}))/m$ over the $m\in\{10, 20,30\}$ agents and $50$ Monte Carlo runs. } 
	\label{fig:con_graphs} \vspace{-15pt}
\end{figure}

\section{Conclusion and Future Work} \label{sec:con}
In this work, we explored the properties of non-Bayesian social learning with Gaussian uncertain models, where the amount of prior evidence collected to estimate the mean and variance of the statistical models may vary between $0$ and $\infty$. We built upon the concept of multinomial uncertain models \cite{HULJ2019_TSP} and have concluded that the Gaussian and multinomial uncertain models have the same underlying properties that allow a group of social agents to perform distributed inference. The main difference between the two approaches is that the uncertain likelihood update and the beliefs point of convergence differs. However, this difference does not influence the learning process. 

For future work, we seek to understand the noninformative prior parameters and identify values that enhance inference decisions. We also plan to extend the analysis to other parametric distributions for real-valued measurements and understand for what family of distributions the convergence properties still hold. Finally, we plan to consider non-parametric distributions. 

\appendices
\bibliographystyle{IEEEtran}
\bibliography{ref}

\end{document}